\documentclass[10pt,twocolumn ]{IEEEtran}
\usepackage{graphicx, color,  amssymb, amsmath, cite}
\usepackage{amsthm}
\usepackage[english]{babel}
\usepackage{cuted}
\usepackage{framed}

\makeatletter
\newcommand{\subjclass}[2][1991]{%
  \let\@oldtitle\@title%
  \gdef\@title{\@oldtitle\footnotetext{#1 \emph{Mathematics subject classification.} #2}}%
}

\makeatother

\newtheorem{thrm}{Theorem}[section]
\newtheorem{prop}[thrm]{Proposition}
\newtheorem{lem}[thrm]{Lemma}
\newtheorem{cor}[thrm]{Corollary}
\newtheorem{defn}[thrm]{Definition}

\newtheorem{ex}[thrm]{Example}
\newtheorem{conj}[thrm]{Conjecture}

\title{Collection and Dissemination of Data on Time-Varying Digraphs}

\author{ Kevin Topley \\  Email: kevint@ece.ubc.ca}

\date{}
\begin{document}

\maketitle


\begin{abstract} Given a network of fixed size $n$ and an initial distribution of data, we derive sufficient connectivity conditions on a sequence of time-varying digraphs for (a) data collection and (b)  data dissemination, within at most $(n-1)$ iterations. The former is shown to enable distributed computation of the network size $n$, while the latter does not. Knowledge of $n$ subsequently enables each node to acknowledge the earliest time point at which they can cease communication, specifically we find the number of redundant signals can be truncated at the finite time $n$. Using a probabilistic approach, we obtain tight upper and lower bounds for the expected time until the \emph{last} node obtains the entire collection of data, in other words complete data dissemination. Similarly tight upper and lower bounds are also found for the expected time until the \emph{first} node obtains the entire collection of data. Interestingly, these bounds are both $\Theta (\text{log}_2(n))$ and in fact differ by only two iterations. Numerical results are explored and verify each result. \end{abstract}

\section{Introduction}\label{sec:intro}

For dynamic communication networks, there may be a given routing algorithm that is optimal for a specific dissemination problem and particular mobility model of the network. Any such routing algorithm is described first by its design, and then is tested via simulations that employ the given assumptions of the network dynamics. When mobility constraints are ``coarsely-grained", that is, occur on a slow time-scale, then the network is relatively stable and thus certain routing protocols will allow for efficient search and usage of paths between the source and destination \cite{smm}. Unfortunately, this process is likely to lead only slowly to insights into the fundamental principles underlying mobile computing \cite{term}. We are thus interested to what extent a high level of mobility can challenge the ability of a network to collect and disseminate data.

Distributed algorithms for data collection and dissemination on static graphs mainly focus on the issue of locality, in other words they operate based on the direct neighborhood of each node \cite{ned, ol04}. In that context, a natural question that may arise is what exactly can be computed locally given limited node resources and/or number of rounds of communication \cite{moni}. Similarly, if a network is highly mobile, then the information being transmitted or received by any node is also restricted to its local, current neighborhood, particularly because the graph might change too quickly to gather any more information. In this sense, we are not so much interested in the message complexity of the task, but rather of its certainty to be solvable in the first place.

We will consider directed graphs that are time-varying and also, to varying degrees, unconstrained in regard to network connectivity. In specific terms, we allow the network digraph to change at each iteration, presuming the set of messages sent at the previous iteration have been received. This in essence is synonymous with synchronous communication; generalizations to analogous asynchronous models follow without any greater insight of our results. To put it another way, we make no assumptions about how long it takes for a message to transfer or how long a path remains stable, only that the network graph locally cannot change faster than it takes for a single message to be delivered. This has been termed ``fine-grained" mobility \cite{term}, which implies that we let the mobility of the nodes be fast enough such that there is no guarantee for a node to send a query to its neighbor and wait for an answer, and slow enough such that when a new neighbor arrives, it will be able to receive a message from its new neighborhood. We do not require that any node is able to detect when its neighborhood changes. Despite this, we are able to solve the iconic ``two-army problem", but only when assuming the disconnectivity conditions of the network are sufficiently constrained. In our deterministic approach we consider varying degrees of network disconnectivity; in our probabilistic approach we assume that the digraph is connected at all times, which, due to our deterministic results, is sufficient for the distributed algorithm to surely succeed after $(n-1)$ iterations.

As a distributed algorithm, we focus mainly on flooding and the specific case of flooding referred to as routing. We consider routing to be the basic problem of a network where information needs to be transported from all nodes to a unique sink node, which concatenates the data. Conversely, routing may require a unique source node to disseminate its local data to all other nodes in the network. As the network changes, any particular routing algorithm can become redundant and ineffective as a method of data collection or dissemination, therefore flooding is an effective remedy to the routing algorithm insofar as it is able to find a path to a set of possibly numerous destinations. Flooding is also of interest when all nodes wish to distribute their local information to the rest of the entire network. This task results in a correct consensus of all nodes regardless of the particular distributed problem that is to be solved (i.e., \cite{olsh,kt11,SB05,bel}). We obtain non-trivial connectivity conditions on when this can be achieved if every node simply broadcasts its set of collected messages, once received, to its neighbors at all times. Furthermore, since each node has limited resources at its disposal, we provide a time point of termination at which all nodes can cease transmitting data; in the deterministic case this time point is less than $(n)$, in the probabilistic case we can expect the termination point to be earlier than $2 \big( \text{log}_2(n) \big)$.

\subsection{Related Work}\label{sec:intro2}

A large number of specific protocols have been proposed for efficient routing in mobile networks \cite{daniel}. These algorithms are based on the available topological knowledge of the network communication graph, or sometimes on the relative coordinates of each node. For the latter set of ``geometric" routing algorithms, it has been shown that greedy (i.e., local) forwarding strategies may lead into dead-ends, while the optimal delivery strategy has only been guaranteed for the static case \cite{fabian} because it involves a preprocessing stage (see also \cite{li,mo08,sun,kok,dim} for further developments in this area). Considering only topology-based routing algorithms, reactive protocols similar to DSR \cite{david} and AODV \cite{charles} make sense in a highly mobile environment. An alternative approach FRESH \cite{henri} takes mobility into account, but, arguably for the worse, it views mobility as a resource rather than a handicap. There also exist hybrid protocols, such as IZR \cite{prince}. Nevertheless, none of these protocols take into account both a potential topological change at every single iteration, and simultaneously a worst-case perspective.

For a sequence of time-varying digraphs, the only thing each node can know with certainty about the graph is the set of incoming messages from their direct neighbors at each time step. Thus we can claim that algorithms designed for dynamic graphs are related to local graph algorithms with respect to the complexity of broadcasts in the network and the limited message sizes/knowledge sets at each node \cite{peleg}. If nodes do not know their neighborhoods, then any broadcast algorithm needs messages in the order of the number of edges in the graph (counting one message per edge).  Various such other lower bounds are given in \cite{faith}. However, these apply only to static graphs; in dynamic graphs, the first question to address is that of solvability. We will answer this problem via network connectivity conditions which provably cannot be improved upon.

\subsection{Outline}\label{sec:intro3}
 
The model of network mobility and the set of algorithmic goals are introduced in Sec.$\ref{sec:model}$. In Sec.$\ref{sec:goal}$ we define the ``collection problem", ``dissemination problem", ``knowledge problem", and ``termination problem". The main results in the deterministic case are presented in Sec.$\ref{sec:main}$, which includes our particular solutions for each of the aforementioned problems. These solutions are the basic tools which we then apply to our results concerning a probabilistic approach of data collection and dissemination, which are presented in Sec.$\ref{sec:main2}$. Numerical simulations illustrated in Sec.$\ref{sec:sim}$ verify our probabilistic analysis and results. These simulations confirm the empirical tightness of our results in terms of the expected time until complete data dissemination and data collection. The conclusions of our results are summarized in Sec.$\ref{sec:con1}$, and some avenues of potential future work are discussed in Sec.$\ref{sec:con2}$. The proofs of all results are contained in the Appendix (Sec. $\ref{sec:app}$). 

\section{Model Assumptions and Algorithmic Goals}\label{sec:model}

Consider a network $\mathcal{V}$ of $n$ nodes, $\mathcal{V} = \{ 1 , 2, \ldots, n \}$, where each node $i \in \mathcal{V}$ has initial data $d_i \in \mathbb{R}$. Inter-node communication is assumed to be instantaneous and occur at discrete times $k \in \mathbb{N}$. At each time instant $k$, the network communication is defined by a time-varying digraph $G(k) = \{  \mathcal{V} , \mathcal{E}(k) \}$, with vertex set $\mathcal{V}$, and edge set $\mathcal{E}(k)$ consisting of ordered pairs of vertices. Node $i$ transmits a message to node $j$ at time $k$ if and only if $(i,j) \in \mathcal{E}(k)$.

Define the ``knowledge state" of node $i$ as $\mathcal{K}_i(k) \supseteq \{ d_i \}$ at time $k$. At $k = 0$ initialize $\mathcal{K}_i(0) = \{d_i \}$ for all $i \in \mathcal{V}$, and assume $d_i \neq  d_j \ \forall \ i \neq j$. If node $i$ transmits a signal to node $j$ at time $k$, we assume that all data held at node $i$ is then contained at node $j$ by time $(k+1)$, \begin{equation}\label{flood}  ( i, j ) \in \mathcal{E}(k) \ \Leftrightarrow \ \mathcal{K}_j(k+1) = \mathcal{K}_j(k) \bigcup \mathcal{K}_i(k)   \end{equation}  The update $(\ref{flood})$ is tantamount to the ``flooding" algorithm considered in \cite{SB05,meh05,kt11,shah}, or ``broadcasting" as it is referred to in, for example, \cite{term,dim}. Implicitly $(\ref{flood})$ states that the local processing time of any signal is negligible, and generally speaking the graph cannot locally change faster than it takes for a message to be transmitted. This type of distributed protocol is not only ostensibly tangible, easily extended to asynchronous communication models, but also has been motivated in the past, for instance, by the shared medium of wireless networks. Note that we do not require any node to be able to learn their actual neighbors at any time, nor do they even require to acknowledge a change in their neighborhood.

\subsection{Algorithmic Goals}\label{sec:goal}

We say that node $j$ has obtained a ``full collection state" (FCS) at time $k$ if $\mathcal{K}_j(k) = \{ d_i \ : \ i \in \mathcal{V} \}$. Conversely, if $\mathcal{V} = \{ i \in \mathcal{V} \ : \ \mathcal{K}_i(k) \supseteq \{ d_j \} \}$, we say that node $j$ has obtained a ``full disseminated state" (FDS) at time $k$. It follows that if at time $k$ all nodes $j \in \mathcal{V}$ have obtained either a FCS or FDS, then the network has obtained a ``full knowledge state" (FKS) at time $k$. We refer the first two of these network connectivity problems respectively as the ``collection problem" and ``dissemination problem".

\begin{defn}\label{def1} Collection Problem (CP): the data $ \{ d_i \}$ of all nodes $i \in \mathcal{V}$ must be obtained by a specific node $q \in \mathcal{V}$.  \end{defn}
  
\begin{defn}\label{def2} Dissemination Problem (DP): the data $\{ d_w \}$ of a specific node $w \in \mathcal{V}$ must be obtained by all nodes $i \in \mathcal{V}$.  \end{defn}

The third problem, which is clearly the composite of CP and DP, is referred to as the ``knowledge problem".

\begin{defn}\label{def3} Knowledge Problem (KP): the data $\{ d_w \}$ of all nodes $w \in \mathcal{V}$ must be obtained by all nodes $i \in \mathcal{V}$.  \end{defn}

Note that, due to $(\ref{flood})$, the solution to the above $3$ problems can be parameterized entirely by network connectivity conditions. This is only one advantage of the flooding algorithm $(\ref{flood})$. Another advantage is that, once a node $i \in \mathcal{V}$ reaches FCS, it can compute \emph{any} function of the initial set of data. The FCS condition is the \emph{only} knowledge state that allows such a privilege. Conversely, if FDS is obtained with respect to a node $j \in \mathcal{V}$, than \emph{all} nodes in the network $\mathcal{V}$ can compute any function of the data $d_j$, and thus the network will have obtained a consensus on this function (presuming all nodes know the common function that they are required to solve). Lastly, if FKS is obtained then we have a network consensus on any function of the initial data \footnote{Again we presume all nodes know the common function that they are required to solve; however, if the initial data itself is able to convey the desired common function, this problem does not require any common \emph{a priori} knowledge within the network, rather it is simply a ``meta-problem" that can solved in an identical fashion to the $3$ problems already defined.}. In terms of distributed computation, there is no knowledge state that is superior to FKS; however, we must assume that node storage and transmission resources are of order (n), or perhaps even larger, depending on what the initial data is and how efficiently it be can be encoded  \cite{SB05,bel,kt11b}.

In addition to the data dissemination and collection problems defined above, a fourth condition can be applied to all $3$. That is of achieving a ``termination" state, wherein every node within the network ceases to transmit signals, thus eliminating the possibility of continuously increasing redundant communication costs.

\begin{defn}\label{def4} Termination Problem (TP): once the problem of interest (CP, DP, or KP) is solved, every node in the network stops transmitting any further messages. \end{defn}

The TP is clearly tantamount to the ``two-army problem", and our solution requires certain connectivity conditions, without which the problem stands unresolved. We will show that the solutions to CP and KP both allow for distributed computation of the network size $n$. This in turn allows each node to know when to cease transmitting signals and thus solves TP when presuming the respective network connectivity conditions hold. Conversely, the solution to DP does not permit a distributed computation of the network size, and thus without \emph{a priori} knowledge of the value of $n$ (or an upper bound for $n$), the connectivity conditions that solve DP will not solve TP, thus no node will know when DP is solved, and hence no node will cease transmitting signals.

\section{Main Results: Deterministic Communication}\label{sec:main}

In this section we present network connectivity conditions that respectively solve CP, DP, KP, and TP (cf. Def.$\ref{def1}-\ref{def4}$) within the finite time $(n-1)$. The connectivity conditions cannot be weakened without increasing the upper bound of $(n-1)$, thus they are not only valid solutions, but also the least restrictive. In Sec.$\ref{sec:main2}$ we consider sequences of random digraphs, each of which is assumed to be connected, thus satisfying all the conditions that solve each of the $4$ aforementioned problems. In Sec.$\ref{sec:main2}$ we prove that the upper bound $(n-1)$ reduces \emph{in expectation} to a function of $n$ that is less than $2  \big( \text{log}_2(n) \big)$. Nonetheless, in the Appendix it shown by example that the upper bound $(n-1)$ cannot be improved upon, even for sequences of connected random digraphs. We now proceed by introducing various types of connectivity conditions that will be used to present our main results.

\subsection{Network Connectivity Definitions: CP, KP, TP}\label{sec:mainintro}

 Let $\mathcal{V}_{-i} = \mathcal{V} \setminus \{i\}$ for any node $i \in \mathcal{V}$. At any time $k$, define an ``input-cord" to node $i$ as an ordered set of nodes $\mathcal{I} ^i (k) \subseteq \mathcal{V}_{-i}$ with the following properties, \begin{equation}\label{ic} \begin{array}{llll} &  \mbox{(a)} \ \  \mathcal{I}^i _j(k) \neq \mathcal{I}^i _r(k) \  , \ \forall \ r \neq j  \\  &   \mbox{(b)} \ \   ( \mathcal{I}^i _j(k) , \mathcal{I} ^i _{j+1}(k) ) \in \mathcal{E}(k)  \ , \\ & \ \ \ \ \ \ \ \ \forall \ j \in \{ 1, 2, \ldots, | \mathcal{I}^i (k) | - 1 \}  \\  &   \mbox{(c)} \ \   (    \mathcal{I}^i_ {| \mathcal{I}^i(k) |} (k) , i ) \in \mathcal{E}(k)  \end{array}  \end{equation} where $\mathcal{I}^i_j(k)$ denotes the $j^{th}$ entry in $\mathcal{I}^i(k)$, and $| \mathcal{I}^i(k)|$ is the cardinality of $\mathcal{I}^i(k)$.  An example of an input-cord to node $6$ is illustrated in Fig.$\ref{fig1}$. 

\begin{figure}[htb] \center \includegraphics[width=0.9 \linewidth]{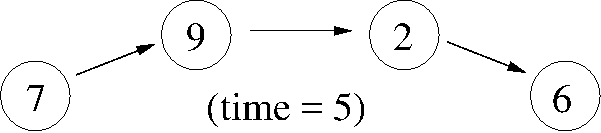} \caption{An input-cord to node $6$, $\mathcal{I}^6(5) = \{ 7,9,2 \}$.} \label{fig1} \end{figure}

We say the input-cord $\mathcal{I}^i(k)$ is ``closed" if $(i, \mathcal{I}_1^i(k) ) \in \mathcal{E}(k)$. The closure of $\mathcal{I}^i(k)$, denoted $ \mathcal{ \hat{I}}^i(k)$, is illustrated in Fig.$\ref{fig4}$.

\begin{figure}[htb] \center \includegraphics[width=0.9 \linewidth]{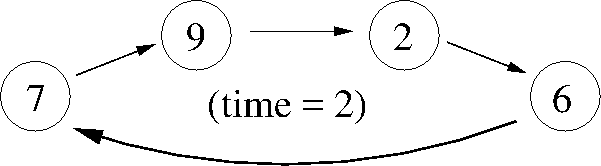} \caption{A closed input-cord to node $6$, $\mathcal{\hat{I}}^6(2) = \{ 7,9,2 \}$.} \label{fig4} \end{figure}

Our first result provides sufficient connectivity conditions to solve CP (cf. Def.$\ref{def1}$).
 
\begin{lem}\label{lem1} Given the update $(\ref{flood})$, if the sequence of input-cords to node $\mbox{q}$ satisfies, \begin{equation}\label{clem1}  | \mathcal{I}^q( k ) | \geq (n-k-1)  \ , \ \forall \ k \in \{ 0,1, \ldots, n-2  \}  \end{equation} then $ | \mathcal{K}_q (k+1) | \geq ( k+2) $. \end{lem}

Given $(\ref{clem1})$, the Lem.$\ref{lem1}$ implies that after $(n-1)$ iterations the node $q$ will know all data in the network, that is $\mathcal{K}_q(n-1) = \{ d_j \ : \ j \in \mathcal{V} \}$, thus solving CP (cf. Def.$\ref{def1}$).

However, the Lem.$\ref{lem1}$ states a stronger property of $\mathcal{K}_q(k)$ than simply $\mathcal{K}_q(n-1) = \{ d_j \ : \ j \in \mathcal{V} \}$. We will utilize Lem.$\ref{lem1}$ more completely to illustrate how the update protocol $(\ref{flood})$ permits distributive computation of the network size $n$. To do so, we require the following definition.

\begin{defn}\label{defcycle}  $\chi (k)$-cycle: if node $i$ has a closed input-cord $\mathcal{\hat{I}}^i(k)$ with cardinality greater than $(\chi - 2)$, then node $i$ is contained in a $\chi(k)$-cycle. The graph $G(k)$ contains $\chi(k)$ if all nodes $i \in 
\mathcal{V}$ are contained in a $\chi(k)$-cycle. \end{defn}

An illustration of a graph that contains $3(4)$ is illustrated in Fig.$\ref{fig5}$.

\begin{figure}[htb] \center \includegraphics[width=0.9 \linewidth]{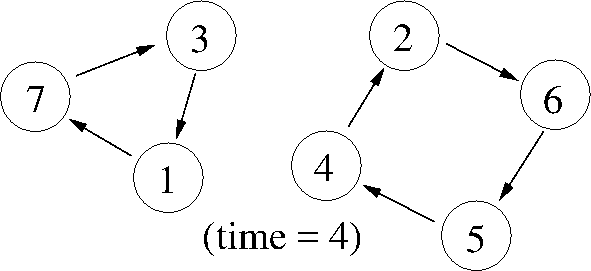} \caption{A graph $G(4)$ that contains $3(4)$. Note that each node $i \in \mathcal{V}$ has a closed input-cord with size greater than $1$.} \label{fig5} \end{figure}

Given the above definition, we now apply Lem.$\ref{lem1}$ to all nodes $i \in \mathcal{V}$, rather than just to node $q$, and obtain the following solution to KP (cf. Def.$\ref{def3}$).

\begin{thrm}\label{thrm1} Define $\psi(k)$ as follows, \begin{equation}\label{psi} \psi (k) = \begin{cases}   n & \text{if } k \leq  \lceil n/2 \rceil - 1 \\   n -  k    & \text{if } k \geq \lceil n/2 \rceil   \end{cases} \end{equation}
 
\noindent  Given the update $(\ref{flood})$, if $\psi(k) \in G(k)$ (cf. Def.$\ref{defcycle}$) for $k \in \{ 0, 1, \ldots, n-2 \}$ then $| \mathcal{K}_i(k+1) | \geq (k+2)$ for all $i \in \mathcal{V}$. \end{thrm}

The next corollary shows that Thm.$\ref{thrm1}$ allows for distributed computation of the network size $n$, and thus solves TP (cf. Def.$\ref{def4}$).

\begin{cor}\label{cor1} Given the update $(\ref{flood})$, if $\psi(k) \in G(k)$  (cf. $(\ref{psi})$)  for $k \in \{ 0, 1, \ldots, n-2 \}$ then at time $k = n$ all nodes $i \in \mathcal{V}$ will know the size of the network $|\mathcal{K}_i(n)| = k = |\mathcal{V}| = n$.   \end{cor}

The connectivity condition defined in Thm.$\ref{thrm1}$ thus not only guarantees FKS by time $(n-1)$, that is $\mathcal{K}_i(n-1) = \{ d_j \ : \ j \in \mathcal{V} \}$ for all $i \in \mathcal{V}$, but also permits distributed computation of the network size \footnote{ Note that, although the connectivity condition in Thm.$\ref{thrm1}$ depends functionally on the network size $n$, it is not required that any node specifically knows the network size.}.

We now move on to the dissemination problem (cf. Def.$\ref{def2}$). Although the connectivity conditions that solve DP are time-symmetric to the conditions that solve CP, when they are applied to the entire network the distributed nature of the dissemination problem does not permit computation of the network size. However, when combining the solution to CP (Lem.$\ref{lem1}$) and the solution to DP presented next (Lem.$\ref{lem2}$), we obtain a solution to KP that assumes weaker connectivity conditions than  Thm.$\ref{thrm1}$. However the Cor.$\ref{cor1}$ cannot be similarly improved upon (see Conj.$\ref{conj1})$.

\subsection{Network Connectivity Definitions: DP, KP, TP}\label{sec:mainintro2}

Analogous to the notion of an ``input-cord", we next define an ``output-cord" from node $i$ as an ordered set of nodes $\mathcal{O} ^i (k) \subseteq \mathcal{V}_{-i}$ with the following properties, \begin{equation}\label{oc} \begin{array}{llll} &  \mbox{(a)} \ \  \mathcal{O}^i _j(k) \neq \mathcal{O}^i _r(k) \  , \ \forall \ r \neq j  \\  &   \mbox{(b)} \ \   ( \mathcal{O}^i _{j+1} (k) , \mathcal{O} ^i _{j}(k) ) \in \mathcal{E}(k)  \ , \\ & \ \ \ \ \ \ \ \ \forall \ j \in \{ 1, 2, \ldots, | \mathcal{O}^i (k) | - 1 \}  \\  &   \mbox{(c)} \ \   (   i ,  \mathcal{O}^i_ {| \mathcal{O}^i(k) |} (k)  ) \in \mathcal{E}(k)  \end{array}  \end{equation} where $\mathcal{O}^i_j(k)$ denotes the $j^{th}$ entry in $\mathcal{O}^i(k)$, and $| \mathcal{O}^i(k)|$ is the cardinality of $\mathcal{O}^i(k)$.  An example of an output-cord from node $6$ is illustrated in Fig.$\ref{fig3}$. 

\begin{figure}[htb] \center \includegraphics[width=0.9 \linewidth]{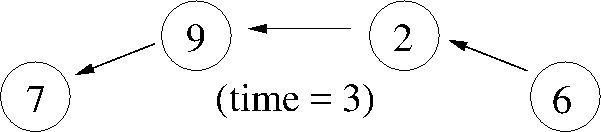} \caption{An output-cord from node $6$, $\mathcal{O}^6(3) = \{ 7,9,2 \}$.} \label{fig3} \end{figure}

Our next result provides sufficient connectivity conditions to solve DP (cf. Def.$\ref{def2}$).

\begin{lem}\label{lem2} Given the update $(\ref{flood})$, if the sequence of ouput-cords from node $\mbox{w}$ satisfies,  \begin{equation}\label{clem2} | \mathcal{O}^w( k ) | \geq (k+1)  \ , \ \forall \ k \in \{ 0,1, \ldots, n-2 \}  \end{equation}  then $ |  \{  \mathcal{K}_i  (k+1)  \supseteq  d_w \ : \ i \in \mathcal{V} \} | \geq (k + 2)$. \end{lem} 

Given $(\ref{clem2})$, the Lem.$\ref{lem2}$ implies that after $(n-1)$ iterations all nodes in the network will know the data $d_w$, that is $ | \{ \mathcal{K}_i(n-1)  \supseteq  d_w \ : \ i \in \mathcal{V} \} | = n$, thus solving DP (cf. Def.$\ref{def2}$). Applying Lem.$\ref{lem2}$ to all nodes $i \in \mathcal{V}$, rather than just node $w$, results in our next theorem, which is an alternative to Thm.$\ref{thrm1}$ as a solution to KP (cf. Def.$\ref{def3})$.



\begin{thrm}\label{thrm2} Define $\nu(k)$ as follows, \begin{equation}\label{nu} \nu (k) = \begin{cases}  k+2   & \text{if } k \leq  \lceil n/2 \rceil - 1 \\   n    & \text{if } k \geq \lceil n/2 \rceil  \end{cases} \end{equation}

\noindent  Given the update $(\ref{flood})$, if $\nu (k) \in G(k)$ (cf. Def.$\ref{defcycle}$) for $k \in \{ 0, 1, \ldots, n-2 \}$ then $| \{ \mathcal{K}_i(k+1)  \supseteq d_w  \ : \ i \in \mathcal{V} \} | \geq (k+2)$ for all $w \in \mathcal{V}$. \end{thrm}

By time $(n-1)$, note that both Thm.$\ref{thrm1}$ and Thm.$\ref{thrm2}$ guarantee FKS, that is, the knowledge set of each node $i \in \mathcal{V}$ contains all of the initial data $\{ d_j \ : \ j \in \mathcal{V}\}$. This implies that (a) the collection problem (cf. Def.$\ref{def1}$) is solved for all nodes $q \in \mathcal{V}$, and (b) the dissemination problem (cf. Def.$\ref{def2}$) is solved for all nodes $w \in \mathcal{V}$. By combining these two theorems, we obtain a significantly weaker sufficient condition for FKS by time $(n-1)$. The drawback of this result, similar to that of Thm.$\ref{thrm2}$, is that it does not permit distributed computation of the network size (in the Appendix this is conjectured, as we foresee no way in which it can or cannot be proven).

\begin{thrm}\label{thrm3} Let $\eta(k) = \mathrm{min} \{ \psi(k), \nu(k) \}$. Given the update $(\ref{flood})$, if $\eta (k) \in G(k)$ (cf. Def.$\ref{defcycle}$) for $k \in \{ 0, 1, \ldots, n-2 \}$ then $ \sum_{i = 1}^n |  \mathcal{K}_i(n-1) | = n^2$. \end{thrm}

The connectivity conditions for Thm.$\ref{thrm1}, \ref{thrm2}$ and $\ref{thrm3}$ are plotted in Fig.$\ref{fig6}$. Note that $\eta(k)$ is significantly lower than $\psi(k)$ for $k \in \{ 0,1, \ldots, ( n/2 )- 2 \}$, and significantly lower than $\nu(k)$ for $k \in \{  n/2  , ( n/2 ) + 1, \ldots, n-2 \}$. At $k = (n/2) - 1$ Thm.$\ref{thrm3}$ requires $n(k) \in G(k)$, thus all nodes are contained in each other node's input-cord. This condition, known as a ``connected" network graph, is illustrated in Fig.$\ref{fig7}$.

 \begin{figure}[htb] \center \includegraphics[width=1 \linewidth]{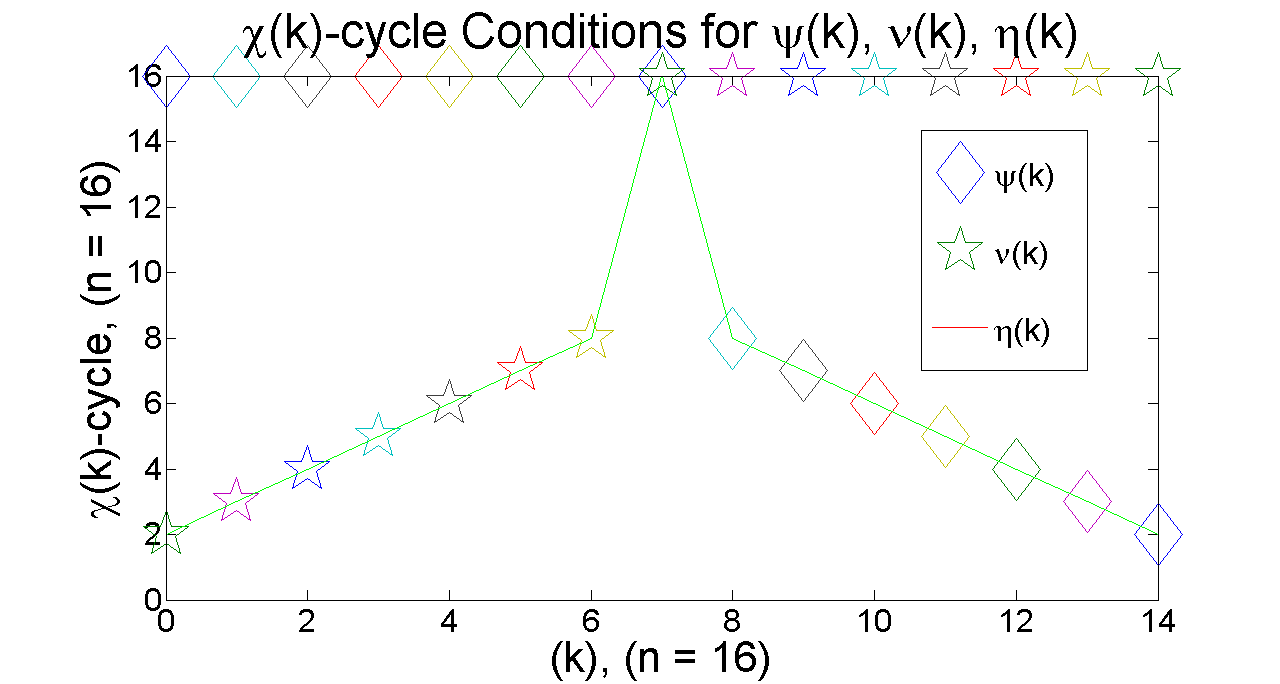} \caption{The $\chi(k)$-cycle conditions for Thm.$\ref{thrm1}, \ref{thrm2},\ref{thrm3}$, plotted respectively as $\psi(k), \nu(k) , \eta(k)$.}\label{fig6} \end{figure}

\begin{figure}[htb] \center \includegraphics[width=0.9 \linewidth]{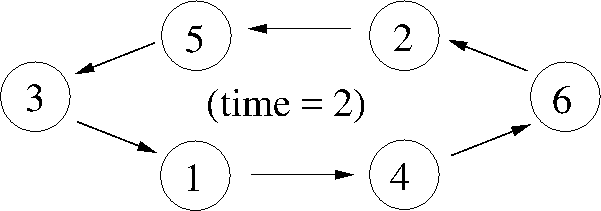} \caption{A connected digraph, $n(2) = 6(2) \in G(2)$.} \label{fig7} \end{figure}

\section{Main Results: Probabilistic Communication}\label{sec:main2}

For a fixed network size $n$, define the set of permutations of $\mathcal{V} = \{ 1, 2, \ldots, n \}$ as $\text{p} ( \mathcal{V} )$. The set  $\text{p} ( \mathcal{V} )$ has cardinality $n!$. Denote the $\ell^{th}$ element of $\text{p} ( \mathcal{V} )$ as $\text{p} ^ \ell ( \mathcal{V} )$, and $j^{th}$ element of $\text{p} ^ \ell ( \mathcal{V} )$, as $\text{p} ^ \ell _ j ( \mathcal{V} )$. At time $k$ let $r(k)$ be a random integer from the set $\{ 1 , 2, \ldots, n! \}$. We define a random connected graph on $n$ vertices as a digraph $\tilde{G}(k) = \{ \mathcal{V} , \tilde{\mathcal{E}}(k) \}$  with edge set $\mathcal{\tilde{E}}(k)$,   \begin{equation}\label{probg} \begin{array}{llll} &  ( \text{p} ^{r(k)} _ j ( \mathcal{V} ) , \text{p} ^{r(k)} _{j+1} ( \mathcal{V} )) \in \tilde{\mathcal{E}}(k) \  , \ \forall \  j \in \mathcal{V}_{-n}   \\ &   ( \text{p} ^{r(k)} _ n ( \mathcal{V} ) , \text{p} ^{r(k)} _1 ( \mathcal{V} )) \in \tilde{\mathcal{E}}(k)  \  . \end{array} \end{equation} Note that $n(k) \in \tilde{G}(k)$ (cf. Def.$\ref{defcycle}$).

Let $\mathcal{R}(k)$ equal a sequence of $k$ random graphs, $\mathcal{R}(k) = \{  \tilde{G}(0) , \tilde{G}(1) , \ldots, \tilde{G}(k-1) \}$. If we define $\mathbb{E}( \cdot)$ as the expectation operator, then the expected time $\hat{k}(n)$ at which the knowledge set of all nodes equals $\{ d_i   \ : \ i \in \mathcal{V} \}$ is, \begin{equation}\label{hatk} \hat{k}(n) = \text{min} \Big\{ k \in \mathbb{R} \ : \ \mathbb{E} \big( \sum_{i = 1}^n | \mathcal{K}_i (k) |  \big)  = n^2    \Big\} \end{equation} This is the expected time at which a network with $n$ nodes will reach FKS.

Define $\phi(n)$ as, \begin{equation}\label{phidef}  \phi(n) = - \text{log} \Big( 1 - \frac{\text{log}(n-2)}{\text{log}(n)} \Big) \ ,   \end{equation} assuming all logarithms are base $2$. We will upper and lower bound $\hat{k}(n)$ by a parameter $\hat{\gamma}$, defined as, \begin{equation}\label{nudef}  \hat{\gamma} = \text{min} \Big\{ \gamma \in \mathbb{N} \ : \   \gamma + \phi(n) \leq \hat{k}(n)  \Big\}  \ .    \end{equation} Note that the condition $\gamma + \phi(n) \leq \hat{k}(n)$ states that if $k \geq \hat{k}(n)$ then it is expected that at time $k$ the network has reached FKS. In this sense, it is possible for $\gamma + \phi(n) \leq \hat{k}(n)$ and $\gamma - 1 + \phi(n) \nleqslant \hat{k}(n)$ to both hold without contradiction. The following theorem places a tight $\Theta (\text{log}(n))$ bound on $\hat{k}(n)$.

\begin{thrm}\label{probthm}  Given the update $(\ref{flood})$ and a sequence of $(n-1)$ random graphs $\mathcal{R}(n-1)$, the solution to $(\ref{nudef})$ is $\hat{\gamma} = 2$.   \end{thrm} \noindent Note that the if the parameter $\hat{\gamma}$ cannot be any smaller, than it bounds $\hat{k}(n)$ tightly from both below and above (i.e., within a single iteration). This is because network communication can only occur at integer time instances $k \in  \mathbb{N}$. We show in Lem.$\ref{lemthm}$ that $\hat{\gamma}$ cannot be reduced.

The Thm.$\ref{thrm1},\ref{thrm2},\ref{thrm3}$ all guarantee that $\sum_{i = 1}^n | \mathcal{K}_i (n-1) |   = n^2$ for any sequence of $(n-1)$ random graphs $\mathcal{R}(n-1)$. The above result shows that, \emph{in expectation}, the network will reach FKS within the exponentially smaller times $\big( \text{log}(n),  \text{log}(n^2) \big)$, see Prop.$\ref{lem1proof}$ in the Appendix.

The value $\hat{k}(n)$ defines the expected time at which the \emph{last} node in the network $\mathcal{V}$ reaches FCS, thus implying the entire network has reached FKS. We next consider $\check{k}(n)$, which will be defined as the expected time at which the \emph{first} node reaches FCS. Just as Thm.$\ref{probthm}$ can be viewed in regard to Thm.$\ref{thrm1},\ref{thrm2},\ref{thrm3}$, the theorem below can be viewed in regard to Lem.$\ref{lem1},\ref{lem2}$.

Similar to $(\ref{hatk})$, let us define $\check{k}(n)$ as follows,  \begin{equation}\label{checkk} \check{k}(n) = \text{min}_{i \in \mathcal{V}} \ \Big\{ k \in \mathbb{R} \ : \  \mathbb{E} \big(  | \mathcal{K}_i (k) |   = n  \big)     \Big\} \end{equation} The value $\check{k}(n)$ is thus the expected time at which the $\emph{first}$ node in $\mathcal{V}$ will obtain FCS \footnote{For convenience, we will, with abuse of notation, label the ``$\emph{first}$ node to obtain FCS", as the ``\emph{earliest} time of FKS". This is analogous to how the ``\emph{last} node to obtain FCS" implies (with no abuse of notation) the ``\emph{latest} time of FKS".}. Similar to $\hat{k}(n)$, we will upper and lower bound $\check{k}(n)$ by a parameter $\check{\gamma}$, defined as, \begin{equation}\label{nudef1}  \check{\gamma} = \text{min} \Big\{ \gamma \in \mathbb{N} \ : \   \gamma + \phi(n) \leq \check{k}(n)  \Big\}  \ .   \end{equation} The following theorem places a tight $\Theta (\text{log}(n))$ bound on $\check{k}(n)$.

\begin{thrm}\label{probthm1}  Given the update $(\ref{flood})$ and a sequence of $(n-1)$ random graphs $\mathcal{R}(n-1)$, the solution to $(\ref{nudef1})$ is $\check{\gamma} = 0$.  \end{thrm} \noindent Again, we emphasize that because network communication occurs at integer time instances $k \in  \mathbb{N}$, if the parameter $\check{\gamma}$ can be proven as the absolute minimum such value, then it bounds $\check{k}(n)$ tightly from both below and above (i.e., within a single iteration). This is indeed proven within the proof of Thm.$\ref{probthm1}$, which is presented in the Appendix (Sec.$\ref{sec:app}$).

Although Thm.$\ref{probthm}$ implies that the expected time of FKS is upper bounded by $3 + \phi(n) < \text{log}(2n^2)$ for any random sequence of connected graphs, this does not diminish the significance of Thm.$\ref{thrm1},\ref{thrm2},\ref{thrm3}$. In the Appendix, by way of Example $\ref{example}$, we show that there are sequences of connected graphs that imply $| \mathcal{K}_i(k) | = k+1$ for all $i \in \mathcal{V}$ and $k \in \{ 0 , 1, \ldots, n-1 \}$, thus yielding FKS only after $(n-1)$ iterations, just as Thm.$\ref{thrm1},\ref{thrm2},\ref{thrm3}$ guarantee. Hence, without changing the required connectivity conditions, the deterministic bounds stated in Thm.$\ref{thrm1},\ref{thrm2},\ref{thrm3}$ cannot be improved upon.

Likewise, the Thm.$\ref{probthm1}$ implies that the expected time of full data dissemination (DP) (cf. Def.$\ref{def2}$), or equivalently the time at which the first node reaches FCS, is upper bounded by $\phi(n) < \text{log}(n^2 /4)$, see Prop.$\ref{lem1proof}$ in the Appendix. In comparison, the Lem.$\ref{lem1},\ref{lem2}$ imply that for any random sequence of connected graphs, the maximum time of complete data collection (resp. full data dissemination) is the exponentially larger time $(n-1)$. However, this does not necessarily diminish the results of Lem.$\ref{lem1},\ref{lem2}$, since the Appendix contains the Example $\ref{example}$ that proves that there are sequences of connected graphs for which CP (cf. Def.$\ref{def1}$) and DP (cf. Def.$\ref{def2}$) are solved only at the time $(n-1)$. It follows then, that without changing the required connectivity conditions, the deterministic upper bounds stated in Lem.$\ref{lem1},\ref{lem2}$ cannot be improved upon.

\section{Numerical Simulations}\label{sec:sim}

It is not hard to show that $\lceil \text{log} (n) \rceil$ is the earliest time for FKS given any sequence of random graphs (see Thm.$\ref{strthm}$ in the Appendix). This implies that $\hat{k}(n) \geq \lceil \text{log} (n) \rceil$. In Fig.$\ref{fig10}$ we plot the bounds $ \phi(n) + \big( \frac{5  \pm 1}{2} \big)$ for network sizes $n \in \{ 3,4, \ldots, 150\}$. On the same graph we plot the time at which FKS occurs when averaged over $1000$ simulations of random graph sequences $\mathcal{R}(n-1)$. It is clear that the lower and upper bounds derived in Thm.$\ref{probthm}$ are a very close approximation to the empirical average time of FKS, particularly they both remain with one iteration of the empirical average. Analogously, we plot in Fig.$\ref{fig10}$ the lower and upper bounds $\phi(n)$ and $(1+\phi(n))$ to the \emph{earliest} time of FKS, which is also plotted on the graph when averaged over $1000$ simulations. Again, it is clear that the empirical average time to \emph{earliest} FKS  remains within one iteration of the bounds derived in Thm.$\ref{probthm1}$.

\begin{figure}[htb] \center \includegraphics[width=0.9 \linewidth]{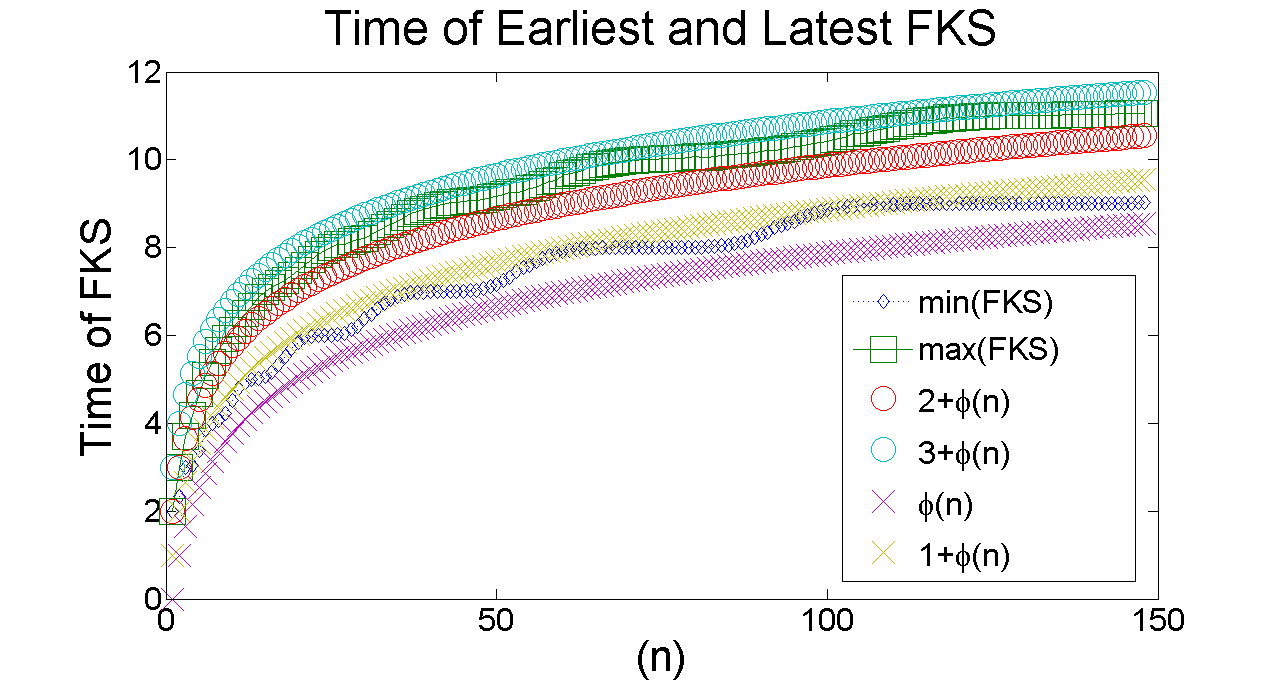} \caption{Empirical time until Earliest FKS and Latest FKS. Network size $\in [3, 150]$. Number of simulations $= 1000$.} \label{fig10} \end{figure}

In Fig.$\ref{fig11}$ we plot the Thm.$\ref{probthm}$ bounds $ \phi(n) + \big( \frac{5  \pm 1}{2} \big)$ for network sizes $n \in \{ 5,4, \ldots, 500\}$. On the same graph we plot the time at which FKS occurs when averaged over $20$ simulations of random graph sequences $\mathcal{R}(n-1)$. Similar to Fig.$\ref{fig10}$, we also plot the bounds of Thm.$\ref{probthm1}$ and the empirical average of the \emph{earliest} time of FKS. The results are identical to Fig.$\ref{fig10}$, despite the fact that we average the times of FKS over 50 times less than Fig.$\ref{fig10}$, and increased the total network size by a factor of $4/3$.

\begin{figure}[htb] \center \includegraphics[width=0.9 \linewidth]{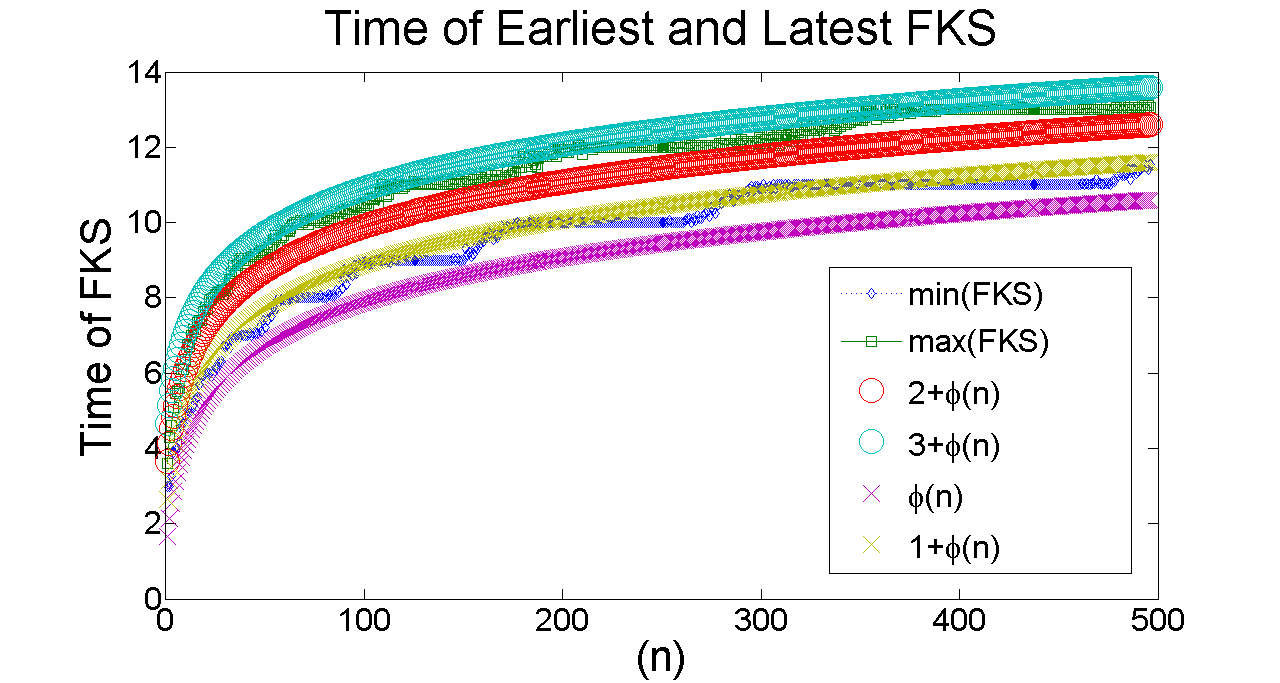} \caption{Empirical time until Earliest FKS and Latest FKS. Network size $\in [5, 500]$. Number of simulations $= 20$.}\label{fig11} \end{figure}

An interesting phenomenon can be seen in both Fig.$\ref{fig10}-\ref{fig11}$, which is the wave-like property of the empirical averaged FKS times. The simulation conditions do not appear to account for this unexpected result, and we can only leave its explanation as potential future work.

\section{Conclusion}\label{sec:con1}

Deterministic connectivity conditions were derived for both data dissemination and data collection on time-varying digraphs. The conditions assumed a ``broadcast" protocol wherein each node sent the entirety of its current data to all nodes within its communication range, which was allowed to vary among nodes and thus requiring the use of digraphs, rather than the far stricter condition of symmetric communication channels. The deterministic conditions were proven using a matrix representation of each communication graph, and non-trivial manipulations thereof. Conditions were given for the ubiquitous termination of all network communication once the particular distribution problem of interest had been solved. The termination problem is tantamount to the ``2-army problem", and, to our knowledge, the conditions given are a unique solution to it.

A probabilistic approach was introduced wherein at each iteration an identically uniform and independently chosen random connected digraph was used to model the communication links for that iteration. Lower and upper bounds were obtained in regard to the expected times of termination for each of the various data distribution problems that were addressed in the deterministic setting. The probabilistic bounds were shown to be exponentially smaller than the upper bounds given by the deterministic results. Nonetheless, examples were given to demonstrate how the deterministic results cannot be improved upon without strengthening the connectivity constraints.

Empirical results confirmed the tightness of the probabilistic bounds, which remained both analytically and empirically within one iteration of the average time to termination, and thus cannot be improved upon since we have assumed communication occurs at discrete time iterations. A curious wave-like property was observed in the empirical average time until termination; an explanation of this remains an open question.

\subsection{Future Work}\label{sec:con2}

Our first concern is regarding the two solutions to KP which do not also solve TP. We predict there may be various connectivity conditions that when applied to the iterations \emph{after} the time $(n-1)$, would lead a termination point within at most the \emph{next} $(n-1)$ iterations. These conditions would likely depend on the two different sets of conditions assumed in Thm.$\ref{thrm2}$ and Thm.$\ref{thrm3}$. Furthermore, these hypothetical conditions would not improve upon the time until FKS, which is already guaranteed by these theorems to occur by time $(n-1)$, but would only yield the benefit of solving TP.  In regard to issue of connectivity conditions beyond the time $(n-1)$, we may consider the case when the initial set of data \emph{changes} with time, in which case none of our results could be directly applied, but rather they would need modifications contingent on the particular dynamics of the local data.

Another source of insight towards efficient data dissemination is that of ordering or classifying different types of network dynamics in terms of the associated redundancy of communications, expediency of data proliferation, and constraints on connectivity. For instance, the sequence of communication graphs presented in Thm.$\ref{strthm}$ yield FKS in minimal time and with minimal redundancy, at the cost of assuming very strict connectivity conditions. Conversely, the Example $\ref{example}$ presents a simple, perhaps overly structured but yet easily implemented sequence of graphs (all of which are identical), that yield FKS in maximum time and with very large (perhaps maximum) communication redundancy. Our probabilistic model can be seen as a mid-point between these two deterministic extremes, as it allows for an independent and identically distributed type of random connected graph at each iteration. The speed of convergence to FKS using this probabilistic model has been quantified and numerically verified in the present work, however the (average) redundancy of communication associated with this model has not been addressed. Quantifying how connectivity conditions relate to communication redundancy would seem to be of worthy consequence, since node storage and communication resources are limited, if not in actuality than at least intuitively, and in some cases expediency to FKS may be trumped by network resources. Even defining an appropriate metric for communication redundancy, and more so for connectivity constraints, remain debatable.

Lastly, it is conceivable that the sequence of connected graphs may not be reducable to a uniformly random element of the set of permutations $\text{p} ( \mathcal{V} )$. Given a different distribution and sample space from which the randomly connected graphs are chosen, the expected times of FKS (resp. FCS, FDS) will certainly vary. It appears to be a considerable task to obtain expectation bounds for FKS when assuming a more general model for the sequence of random communication graphs.

\section{Appendix}\label{sec:app}

\noindent \textbf{Lemma $\ref{lem1}$} \emph{Proof.} The proof of Lem.$\ref{lem1}$ is most easily illustrated by appealing to the proof of Lem.$\ref{lem2}$. The communication conditions assumed in Lem.$\ref{lem2}$ and therein proven to solve DP can be stated as follows in terms of a sequence of edge sets: \begin{equation}\label{array1} \begin{array}{llll}  &  \mathcal{E}(0) \supseteq \{  ( 1 , 2_0 ) \} \\ &  \mathcal{E}(1) \supseteq \{  ( 1 , 2_1 ) ,(  2_1 , 3_1 ) \} \\ & \vdots \\  &  \mathcal{E}( k) \supseteq \{  ( 1 , 2_k )  , \ldots , \\ & \quad \quad \quad \quad \quad \quad \big( (k+1)_k , (k+2)_k \big)  \} \\ & \vdots \\  &  \mathcal{E}(n-2) \supseteq \{  ( 1 , 2_{n-2} )  , \ldots , \\ & \quad \quad \quad \quad \quad  \quad \ \big( (n-1)_{n-2} , n_{n-2} \big)  \}  \end{array} \end{equation} The set of nodes $\{ 1, 2 , \ldots, n \}$ that comprise the network $\mathcal{V}$ is assumed to be invariant with time, thus the sequence of edge sets $\big\{ \mathcal{E}(k) \ : \  k \in \{ 0 , 1 , \ldots, n-2 \} \big\}$ fully defines the sequence of network communication graphs  $\big\{ G(k) \ : \  k \in \{ 0 , 1 , \ldots, n-2 \} \big\}$. Further note that each communication graph $G(k)$ can be defined by a non-negative $(n \times n)$ matrix $M(k) = [ M_{ij}(k) ]$ with positive entries for each pair $(i,j) \in \mathcal{E}(k)$, \begin{equation}\label{mm1}  M_{ij}(k) > 0 \ \ \Leftrightarrow \ \ (i,j) \in \mathcal{E}(k) \ . \end{equation} The update $(\ref{flood})$ implies that the cardinality of any knowledge set $\mathcal{K}_i(k)$ will never decrease, thus we let $M_{ii}(k) > 0$ for all $i \in \mathcal{V}$.

Define $\mathcal{M}(k) \doteq \prod_{r = 0}^{k} M(r)$. In accordance with the update $(\ref{flood})$, the knowledge set of node $i$ will contain data $d_j$ at time $(k+1)$ if and only if $\mathcal{M}_{ji}(k) > 0$,  \begin{equation}\label{mm2}  \mathcal{K}_i(k+1) \supseteq \{ d_j \} \ \ \Leftrightarrow \ \   \mathcal{M}_{ji}(k) > 0 \ . \end{equation} The Lem.$\ref{lem2}$ implies that if $(\ref{array1})$ holds, then $\mathcal{M}_{1j} (n-2) > 0 $  for all $j \in \mathcal{V}$ (note that in Lem.$\ref{lem2}$ we define without loss of generality (WLG) $w = 1$).

To summarize the above: the Lem.$\ref{lem2}$ places conditions on the sequence $ \big\{ \mathcal{E}(k) \ : \ k \in \{ 0 , 1, \ldots, n-2 \} \big\}$ that in turn define the matrices $\big\{ M(k) \ : \ k \in \{0, 1, \ldots, n-2 \} \big\}$, which, when multiplied together, yield $\mathcal{M}_{1j} (n-2) > 0 $  for all $j \in \mathcal{V}$. It thus follows that by transposing the product $\mathcal{M}(k)$ we obtain a sequence of matrices $\big\{ M'(k) = M(n-k-2) \ : \ k \in \{0, 1, \ldots, n-2 \} \big\}$ that in turn define a sequence of edge sets and thus communication graphs.

The CP requires that a single node $q \in \mathcal{V}$ obtain the data that is initially held at every other node $i \in \mathcal{V}_{-q}$. Let $q = n$ WLG. In terms of the matrix $\mathcal{M}(k)$, the CP is solved when $\mathcal{M}{jn} > 0$ for all $ j \in \mathcal{V}$. This condition is simply the transpose of the matrix $\mathcal{M}(n-2)$ once DP is solved, that is $\mathcal{M}_{1j} (n-2) > 0 $  for all $j \in \mathcal{V}$. Thus to solve CP we need only the condition $\big\{ M'(k) = M(n-k-2) \ : \ k \in \{0, 1, \ldots, n-2 \} \big\}$, which, when put in terms of the sequence of edge sets defined in $(\ref{array1})$, is:

\begin{equation}\label{array2} \begin{array}{llll}  &    \mathcal{E}(0) \supseteq \{  ( 1_0 , 2_0), ( 2_0 ,  3_0 ), \ldots, \\ &  \  \ \ \ \  \quad \quad \quad \quad  \quad \quad \quad \quad (  (n-1)_0  , n ) \} \\ & \mathcal{E}(1) \supseteq \{ (2_1, 3_1),  (3_1, 4_1), \ldots ,  \\ &   \  \ \ \ \  \quad \quad \quad \quad \quad \quad \quad \quad  ( (n-1)_1,n )  \} \\ & \vdots \\  &  \mathcal{E}( k) \supseteq \{  \big( (k+1)_k,  (k+2)_k \big)  , \ldots , \\ &   \  \ \quad  \quad  \quad \quad \quad \quad \quad \quad \quad  (  (n-1)_k , n )  \} \\ & \vdots \\  &  \mathcal{E}(n-2) \supseteq \{  ( (n-1)_{n-2}, n )  \}  \end{array} \end{equation} The proof of Lem.$\ref{lem1}$ is complete by noting that the set of conditions $(\ref{array2})$ is identical to the set of conditions $(\ref{clem1})$ defined in Lem.$\ref{lem1}$.  \qed

\noindent \textbf{Theorem $\ref{thrm1}$} \emph{Proof.} If $k \in \{ 0 ,1 , \ldots,  \lceil n/2 \rceil - 1 \}$ then $n(k) \in G(k)$ and thus all nodes have input-cords of length $(n-1)$. It follows that $(\ref{clem1})$ is satisfied, and thus by Lem.$\ref{lem1}$ we have $| \mathcal{K}_i( k+1 ) | \geq  ( k+2)$ for all $k \in \{ 0 ,1 , \ldots,  \lceil n/2 \rceil - 1 \}$ and $i \in \mathcal{V}$. If $k \in \{ \lceil n/2 \rceil ,\lceil n/2 \rceil +1 , \ldots,  n - 2 \}$ then $(n-k)(k) \in G(k)$ which implies each node has an input-cord of at least length $(n-k-1)$. This is precisely the condition $(\ref{clem1})$ required in Lem.$\ref{lem1}$, thus implying $| \mathcal{K}_i( k+1 ) | \geq  ( k+2)$ for all $k \in \{ \lceil n/2 \rceil ,\lceil n/2 \rceil +1 , \ldots,  n - 2  \}$ and $i \in \mathcal{V}$. \qed

\noindent \textbf{Corollary $\ref{cor1}$} \emph{Proof.}   Given the update protocol $(\ref{flood})$ and communication condition $\psi(k) \in G(k)$  (cf. $(\ref{psi})$), the Thm.$\ref{thrm1}$ guarantees that $|\mathcal{K}_i(k) | > k+1$ for all $i \in \mathcal{V}$ and $k \in \{ 0 , 1, \ldots, n-1 \}$.  Each node $i \in\mathcal{V}$ initially holds the unique data $\{d_i\}$, thus we have $|\mathcal{K}_i(k) | \leq n$. It then follows that once $|\mathcal{K}_i(n) | = n$, the node $i$ knows that the condition $|\mathcal{K}_i(k) | > k+1$ no longer holds, and thus $n$ must be the network size. All nodes $i \in \mathcal{V}$ can thus distributively compute the network size $n$ by time $k = n$ under the assumed communication conditions. \qed

\noindent \textbf{Lemma $\ref{lem2}$} \emph{Proof.} Let $w = 1$ WLG. We will assume $| \mathcal{O}^1( k ) | = (k+1)  \ , \ \forall \ k \in \{ 0,1, \ldots, n-2 \}$, since adding more communication links can only increase the quantity of interest $ | \{   \mathcal{K}_i  (k)  \supseteq  d_1 \ : \ i \in \mathcal{V} \} |$, which we intend to show is lower bounded by $(k+1)$.

At time $k=0$, only node $1$ contains the data $d_1$. At $k = 0$ we have $| \mathcal{O}^1( 0 ) | = 1$, so let us label the single node contained in $\mathcal{O}^1( 0 )$ as $2_0 \in \mathcal{V}_{-1}$. Now at time $k = 1$, both nodes $1$ and $2_0$ contain the data $d_1$. At time $k =1$ we have $| \mathcal{O}^1( 1 ) | = 2$, so let us label the two nodes contained in $\mathcal{O}^1 (1 )$ as $2_1 , 3_1 \in \mathcal{V}_{-1}$. If $2_1  = 2_0$, then $3_1 ( \neq 2_0)$ will obtain the data $d_1$ from node $2_1$, and otherwise the node $2_1  ( \neq 2_0)$ will obtain the data $d_1$ from node $1$; in both cases we end up at time $k=2$ with $3$ nodes that contain $d_1$. In summary,  if $ | \{   \mathcal{K}_i  (k)  \supseteq  d_1 \ : \ i \in \mathcal{V} \} | = (k+1)$,  all that is required for  $ | \{   \mathcal{K}_i  (k+1)  \supseteq  d_1 \ : \ i \in \mathcal{V} \} |$ to stay above $(k+2)$ is for some node $i(k)$ that does not contain data $d_1$ to take a place in $\mathcal{O}^1( k )$. Since $| \mathcal{O}^1( k ) | = (k+1)$, the latter condition necessarily must hold.  On the other hand, if $ | \{   \mathcal{K}_i  (k)  \supseteq  d_1 \ : \ i \in \mathcal{V} \} | > (k+1)$, then the condition to be proven already holds, and so, if necessary, we can apply the previous argument at the subsequent time $(k+1)$.  \qed

\noindent \textbf{Theorem $\ref{thrm2}$} \emph{Proof.} If $k \in \{ 0 ,1 , \ldots,  \lceil n/2 \rceil - 1 \}$ then $(k+2)(k) \in G(k)$ and thus all nodes have output-cords of at least length $(k+1)$. This is precisely the condition $(\ref{clem2})$ in Lem.$\ref{lem2}$, and thus $| \{ \mathcal{K}_i(k+1)  \supseteq d_w  \ : \ i \in \mathcal{V} \} | \geq (k+2)$ for all $k \in \{ 0 ,1 , \ldots,  \lceil n/2 \rceil - 1 \}$ and $w \in \mathcal{V}$. If $k \in \{ \lceil n/2 \rceil ,  \ldots,  n - 2 \}$ then $(n)(k) \in G(k)$ which implies each node has an output-cord of length $(n-1)$. It follows that $(\ref{clem2})$ is satisfied and thus by Lem.$\ref{clem2}$ we have $| \{ \mathcal{K}_i(k+1)  \supseteq d_w  \ : \ i \in \mathcal{V} \} | \geq (k+2)$ for all $k \in \{ \lceil n/2 \rceil  ,  \ldots,  n - 2  \}$ and $w \in \mathcal{V}$. \qed

\noindent \textbf{Theorem $\ref{thrm3}$} \emph{Proof.} For $k \in \{ 0, 1, \ldots, \lceil \frac{n}{2} \rceil - 1 \}$ we have $\eta(k) = k+2$ and thus Lem.$\ref{lem2}$ implies $| \{ \mathcal{K}_i(k+1)  \supseteq d_w  \ : \ i \in \mathcal{V} \} | \geq (k+2)$ for all $k \in \{ 0 ,1 , \ldots,  \lceil n/2 \rceil - 1 \}$ and $w \in \mathcal{V}$. By $(\ref{mm1})$ we can define the set of matrices $A \big( \eta(k) \big)  \in \mathbb{R}^{n \times n}$ for which $\eta(k) \in G(k)$. By $(\ref{mm2})$ we then have $\mathcal{M}(\lceil \frac{n}{2} \rceil) = \prod_{r = 0}^{\lceil \frac{n}{2} \rceil - 1 } A \big( \eta(r) \big)$, which by Lem.$\ref{lem2}$ will satisfy, \begin{equation}\label{zero} \sum_{j = 1}^n \mathcal{I} \Big( \mathcal{M} \big( \lceil  n/2 \rceil \big) _{ij} \Big) \geq \lceil n/2  \rceil + 1 \ ,  \end{equation} where $\mathcal{I}(\cdot)$ is the indicator function defined as, $$  \mathcal{I}(A_{ij}) =  \begin{cases}   1      & \quad \text{if } A_{ij} > 0 \\    0   & \quad \text{if }  A_{ij}  = 0  \\  \end{cases} $$ for a matrix $A = [A_{ij}]$.

For $k \in \{   \lceil \frac{n}{2} \rceil , \ldots, n-2  \}$ we have $\eta(k) = n- k$. By $(\ref{mm1})$ this defines the set of matrices $A \big( \eta(k) \big)  \in \mathbb{R}^{n \times n}$ for which $\eta(k) \in G(k)$, $\eta(k) = n- k$ and $k \in \{   \lceil \frac{n}{2} \rceil , \ldots, n-2  \}$.  We now let, $$ \mathcal{\tilde{M}}(n-1) = \prod_{r = \lceil \frac{n}{2} \rceil }^{n-2 } A \big( \eta(r) \big) \ , $$ and seek to prove, \begin{equation}\label{newsum} \textbf{1}_n'  \mathcal{I}  \big(  \mathcal{M} ( \lceil  n/2 \rceil ) \cdot \mathcal{\tilde{M}}(n-1) \big)   \textbf{1}_n  = n^2  \ , \end{equation} where $\textbf{1}_n \in \mathbb{R}^{n \times 1}$ is a vector of unit values. Note that $(\ref{newsum})$ is equivalent to, \begin{equation}\label{newsum1} \textbf{1}_n'  \mathcal{I}  \big(  \mathcal{\tilde{M}}(n-1) ' \cdot \mathcal{M} ( \lceil  n/2 \rceil ) '  \big)  \textbf{1}_n  = n^2  \ . \end{equation}

Consider the transpose of $\mathcal{\tilde{M}}(n-1)$, \begin{equation}\label{newsum2} \begin{array}{llll}  \mathcal{\tilde{M}}(n-1)'  &  = \prod_{r = n-2}^{ \lceil \frac{n}{2} \rfloor  } A ' \big( \eta(r) \big) \\ &  =  \prod_{r = 0}^{ \lceil \frac{n}{2} \rceil - 1} A ' (r+2) \\ &  =  \prod_{r = 0}^{ \lceil \frac{n}{2} \rceil - 1} A(r+2) \\ & = \mathcal{M}( \lceil  n/2 \rceil) \ .  \end{array} \end{equation} From $(\ref{newsum2})$ we obtain, \begin{equation}\label{newsum3} \begin{array}{llll} &  \mathcal{I}  \big(  \mathcal{\tilde{M}}(n-1) ' \cdot \mathcal{M} ( \lceil  n/2 \rceil ) '  \big) _{ij} = \\ &   \sum_{\ell = 1}^n \mathcal{I} \big(  \mathcal{M} ( \lceil  n/2 \rceil )_{i \ell} \mathcal{M} ( \lceil  n/2 \rceil )_{j \ell} \ \big) \ . \end{array} \end{equation} The condition $(\ref{zero})$ implies that, for any $(i,j) \in \mathcal{V}^2$, there must exist at least one value of $\ell \in \mathcal{V}$ for which both $\mathcal{M} ( \lceil  n/2 \rceil )_{i \ell}$ and  $\mathcal{M} ( \lceil  n/2 \rceil )_{j \ell}$ are non-zero, thus implying $(\ref{newsum1})$. \qed

\noindent \textbf{Theorem $\ref{probthm}$} \emph{Proof.} Define $P(k) \in \mathbb{R}^{n \times n}$ as follows $P(k)_{ij}  = \textbf{P}[ \mathcal{K}_{i}(k) \supseteq d_j ]$, where $\textbf{P}[A]$ is the probability that condition $A$ is true. Due to symmetry, the expected value of any sequence of random graphs will yield at time $k$,

$$ 
P(k) = 
 \begin{pmatrix}
  1 & x_k & \cdots & x_k \\
  x_k & 1 & \cdots & x_k \\
  \vdots  & \vdots  & \ddots & \vdots  \\
  x_k & x_k & \cdots & 1 
 \end{pmatrix}
$$ where $x_k$ is the probability that at time $k$ the knowledge set of node $i$ contains the data $d_j$ for any $j \in \mathcal{V}_{-i}$, $$  x_{k} = \textbf{P} [   \mathcal{K}_i(k) \supseteq  d_j    \ : \ i \in \mathcal{V}  \ , \ j \in \mathcal{V}_{-i}  ] \ .  $$  Next we let, $$ \begin{array}{llll} &   x_{k 0} = \textbf{P} \big[  \mathcal{K}_i(k) \nsupseteq d_j    \ : \ i \in \mathcal{V}  \ , \ j \in \mathcal{V}_{-i}   \big]  \\  & \  \ \quad = 1 - x_k \ ,  \end{array} $$ thus $x_{k 0}$ equals the probability that at time $k$ node $i$ does not contain the data $d_j$.

We will denote $E_k = \mathbb{E}(| \mathcal{K}_i (k)  \setminus d_i | )$ for any $i \in \mathcal{V}$, and similarly denote $E ' _k = \mathbb{E}(| \mathcal{K}_i(k) | ) = E_k + 1$. Accordingly, we initially have $E_0 = 0$ and $x_{00} = 1$. The probability $x_{k 0}$ is recursively defined as, $$  \begin{array}{llll} &  x_{k 0} =  \textbf{P} \big[  \mathcal{K}_i(k) \nsupseteq d_j \big] \\ & \ \ \ \ \ =   \textbf{P} \big[  \mathcal{K}_i(k-1) \nsupseteq d_j \big] \cdot \\ & \ \ \textbf{P} \big[  (\ell, i ) \in \mathcal{E}(k-1)  , \ \mathcal{K}_{\ell}(k-1) \nsupseteq d_j \  |  \\ & \quad  \quad  \quad \quad \quad \quad \ \ \ \ \ \ \ \ \ \ \ \ \mathcal{K}_i(k-1) \nsupseteq d_j  \big] \\ &  = x_{(k-1)0} \Big( \frac{n - 1 - E'_k }{n-1} \Big)  =  x_{(k-1)0} \Big( \frac{n - 2- E_k }{n-1} \Big) \ ,  \end{array}  $$ wherefrom the expectation $E_k$ can be defined recursively as, \begin{equation}\label{rec2} \begin{array}{llll}  E_k &  =  (n-1) x_k \\ & = n - 1 - (n-1) x_{k 0} \\ & =   n -  1 -  \frac{ \prod_{r= 0}^{k-1} ( n - 2 - E_r ) }{ (n-1)^{k-1} }  \  . \end{array} \end{equation} Rearranging $(\ref{rec2})$ and taking logarithms yields, $$ \begin{array}{llll} & \text{log} \Big( 1 - \frac{E_k}{n-1} \Big) = k \cdot \text{log} \Big( \frac{n-2}{n-1} \Big)  + \\ & \quad  \quad \quad \quad   \quad    \quad    \quad  \sum_{r = 1}^{k-1} \text{log} \Big( 1 - \frac{E_r}{n-2} \Big) \end{array} $$ from which we obtain, \begin{equation}\label{rec3} \begin{array}{llll}  & \text{log} \Big( 1 - \frac{E_k}{n-2} \Big) <  \text{log} \Big( 1 - \frac{E_k}{n-1} \Big)  = \\ & \ \ \  k \cdot \text{log} \Big( \frac{n-2}{n-1} \Big)  +    \sum_{r = 1}^{k-1} \text{log} \Big( 1 - \frac{E_r}{n-2} \Big)   \\  &   <  \big( k   + \sum_{r = 1}^{k-2} 2^{r-1} (k-r) \big) \cdot \text{log} \Big( \frac{n-2}{n-1} \Big)  \\ & \quad  \quad  \quad  \quad  \quad  \quad  \quad  \quad   + 2^{k-2} \text{log} \Big(  \frac{n-3}{n-2}  \Big) \ .   \end{array} \end{equation} Note that $(\ref{rec3})$ was obtained by iterating $(\ref{rec2})$ backwards until $k = 2$. At the point $k=2$ we used the fact that $E_1 = 1$, which can be obtained from $(\ref{rec2})$ and the initial condition $E_0 = 0$.

Let $\chi_k =  k   + \sum_{r = 1}^{k-2} 2^{r-1} (k-r)$.  Note that $\chi_k = 3 \cdot 2^{k-2}-1$, which we now prove by induction.  Let $k = 3$, then $\chi_3 = 5 = 3 \cdot 2^{3-2} - 1$. Next assume $\chi_k = 3 \cdot 2^{k-2}-1$, and consider $\chi_{k+1}$, 
 $$ \begin{array}{llll} & \chi_{k+1} = k   + 1 + \sum_{r = 1}^{k-1} 2^{r-1} (k-r+1)   \\  &  \ \ \ \ = \chi_k +1 + 2^{k-2} + \sum_{r = 1}^{k-1} 2^{r-1} \\ & \ \ \ \ = \big( 3 \cdot 2^{k-2} - 1 \big)  + 1 + 2^{k-2} + 2^{k-1} - 1 \\ & \ \ \ \ = 3 \cdot 2^{k-1} - 1 \ . \end{array} $$  Solving $(\ref{rec3})$ for $E_k$ yields, \begin{equation}\label{b11}  E_k > (n-1) \Bigg( 1 -  \Big( \frac{n-3}{n-2} \Big)^{2^{k-2}} \Big( \frac{n-2}{n-1} \Big)^{\chi_k}    \Bigg)   \end{equation} which we now show is lower bounded by, \begin{equation}\label{n1} n \Bigg( 1 - \Big( \frac{n-2}{n} \Big)^{k-1} \Bigg) - 1 \ . \end{equation} To lower bound $(\ref{b11})$ we upper bound $\Big( \frac{n-2}{n-1} \Big)^{\chi_k} < \Big( \frac{n-2}{n-1} \Big)^{2^{k-1}}$ which yields, \begin{equation}\label{pp1}  E_k > (n-1) \Bigg( 1 -  \Big( \frac{(n-3)(n-2)}{(n-1)^2} \Big)^{2^{k-2}}   \Bigg) \ . \end{equation} Canceling the equivalent terms in $(\ref{pp1})$ and $(\ref{n1})$ then rearranging it will suffice to show, $$ \begin{array}{llll} &  \frac{(n-1)^2}{(n-2) n } <  \Big( \frac{ (n-1)^{3}}{n^{2} (n-3)} \Big) ^{2^{k-2}}  \ .  \end{array} $$   Assuming $k \geq 2$ we can proceed,  $$ \begin{array}{llll} &  1 + \frac{1}{n^2 - 2 n} < 1  + \frac{3n - 1}{n^3 - 3 n^2} \\ & \therefore  n^3 - 3 n^2  < 3 n^3 - 7 n^2  + 2n  \\ & \therefore 0 < (n-1)^2 \ , \end{array} $$ thus verifying the lower bound $(\ref{n1})$.

Due again to symmetry, we have $\sum_{i = 1}^n E'_k = n E'_k$. We now assume $n E'_k > n^2 - 1$, from which we can then infer $\mathbb{E} \big(  \sum_{i = 1}^n  | \mathcal{K}_i (k) | \big) = n^2$. From $(\ref{n1})$ we have, \begin{equation}\label{rec4}  E'_k > n \Bigg(  1 -  \Big( \frac{n-2}{n}\Big)^{2^{k-1}}    \Bigg)    \ .   \end{equation} We now solve $(\ref{rec4})$ for a sufficiently large $k$ when supposing $E'_k > n - (1/n)$, 
\begin{equation}\label{rec5} \begin{array}{llll} &    n \Bigg(  1 -  \Big( \frac{n-2}{n}\Big)^{2^{k-1}}    \Bigg)  \geq  n - (1/n) \ , \\ &   \therefore \ \ \ \Big( \frac{n-2}{n}\Big)^{2^{k-1}}   \leq \frac{1}{n^2} \ , \\ & \therefore \ \  k \geq 2 + \text{log}  \Big(  \frac{\text{log}(n)}{\text{log} \big(n/(n-2)\big)}     \Big) \ .    \end{array}  \end{equation} Note that, $$ \frac{\text{log}(n)}{\text{log} \big(n/(n-2)\big)}  = 2^{ \phi(n)}  \ , $$ which can be easily verified by $(\ref{phidef})$.  \qed

\vspace{0.4 cm}

The following lemma will be used in the proof of Thm.$\ref{probthm1}$. Note that this lemma, together with Thm.$\ref{probthm}$, shows that the bound $\hat{k}$ cannot be reduced, and thus is tight. In other words, one iteration below $\hat{k}$ and we do not expect FKS, that is $\mathbb{E} \big( \sum_{i = 1}^n | \mathcal{K}_i( k) |  = n^2 \big)  > \hat{k}-1$.

\begin{lem}\label{lemthm}  For $\hat{k} = 2 + \phi(n)$ (cf.$(\ref{phidef})$), $E_{\hat{k}-1} < (n^2- n - 1)/n$.   \end{lem}

\begin{proof} For convenience denote $\gamma_{\phi(n)} = \Big( \frac{n-3}{n-2} \Big)^{2^{\phi(n)}} \Big( \frac{n-2}{n-1} \Big)^{3 \cdot 2^{\phi(n)} -1}$. We denote $\hat{E}_k$ as the maximum value of $E_k$. Applying Thm.$\ref{probthm}$ and $(\ref{b11})$ we have, $$ \begin{array}{llll} &  E_{\hat{k} - 1} - E_{\hat{k}} <   \hat{E}_{\hat{k} - 1}  -  (n-1)(1- \gamma_{\phi(n)}) \\ & \ \ \ \  \ \ \ \ \ \ \ \ \ \ <  \hat{E}_{\hat{k} - 1}  -   \frac{n^2 -n - 1}{n}  \\ &  \therefore \  \  \ E_{\hat{k} - 1} <   E_{\hat{k}}  + \hat{E}_{\hat{k}-1} - n + 1 + n^{-1} \\ &  \ \ \ \ \ \ \ \  \ \ \ \ \ < \hat{E}_{\hat{k}-1}  +  n^{-1}   \  ,  \end{array} $$ where we have used the fact $ E_k < n- 1$ for all $k \in \{0,1, \ldots, n-2 \}$, which holds due to Example $\ref{example}$ (a fixed connected digraph). Upper bounding the right-hand side (RHS) of the above by $(n - 1 - n^{-1})$ and solving for $\hat{E}_{\hat{k}-1}$ yields the condition, $$ \hat{E}_{\hat{k}-1} <  n - 1 - 2 n^{-1} \ . $$

Notice that the upper bound $(n - 1 - n^{-1})$ was chosen without any specific precondition. For this reason, along with the fact that $E_{\ell} < E_{\ell+1}$ for all $\ell \in \{ 0 , 1, \ldots, n-2 \}$ (which can be inferred simply by noting that more signals will have occurred by time $(\ell+1)$ as compared to time $\ell$), we must confirm the following, $$ \begin{array}{llll} &  \hat{E}_{\hat{k} - 1} - E_{\hat{k}} <   \hat{E}_{\hat{k} - 1}  -  (n-1)(1- \gamma_{\phi(n)}) \\ & \ \ < n-1 -2 n^{-1} - (n-1)(1- \gamma_{\phi(n)} ) \\ & \ \  =  -2 n^{-1} + (n-1) \gamma_{\phi(n)} < 0\ . \end{array} $$ Rearranging the last inequality and taking logarithms yields,  \begin{equation}\label{j13}  \frac{\text{log}(n)}{   \text{log} \Big( \frac{n}{n-2}  \Big)  } > \frac{   \text{log}\Big( \frac{n (n-1)^2 }{ 2 (n-2)  } \Big) } {  \text{log}\Big( \frac{(n-1)^3}{(n-2)^2 (n-3)} \Big) } \ .   \end{equation}

For notational convenience let $a = \text{log}(n)$, $b = \text{log}(n-2)$, $c = \text{log} \Big( \frac{(n-1)^3}{n-3} \Big)$, and $d = \text{log}\Big( \frac{(n-1)^2}{2} \Big)$.
Note that $a > b > 0$, $c > 2b$. Rearranging $(\ref{j13})$ yields, $$ \begin{array}{llll} & \ \  \frac{a}{a - b} > \frac{ a- b + d}{c - 2b}  \\ & \therefore  \ \ a c >  a^2+ b^2 + \big( 2  \cdot \text{log}(n-1)-1 \big)(a-b) \\ & \therefore \ \ \  c >  a + \frac{ b^2 }{a} + 2  \cdot \text{log}(n-1)-1   \\ & \ \ \ \ \ \  \ \ \ - 2 (b/a) \cdot \text{log}(n-1) + (b/a)  \ . \end{array} $$ Upper bounding the last $(b/a)$ term by $1$, replacing $a,b$ and $c$ with their definitions, and combining terms we obtain, $$\text{log} \Big(  \frac{(n-1)^3}{n-3}\Big) > \text{log} \Bigg(   \frac{ (n-2)^2 (n-1)^2 }{(n-1)^{2 \frac{\text{log}(n-2)}{\text{log}(n)}}}  \Bigg)  \ .  $$ Canceling the logarithms and rearranging yields, $$  (n-1)^{1+ 2 \cdot \frac{\text{log}(n-2)}{\text{log}(n)}} > (n-3) (n-2)^2 \ . $$  Taking logarithms and rearranging then yields, $$ \begin{array}{llll} & 2 \cdot \text{log} ( n-2)  \text{log} (n-1) +  \text{log} ( n-1)  \text{log} (n)  \\ &  \ \ \ >  \text{log} ( n-3)\text{log} ( n) +  2 \cdot \text{log} ( n-2 ) \text{log} ( n) \  . \end{array} $$ Concavity of the logarithm function implies $\text{log} ( n-2) / \text{log} ( n-3) >   \text{log} ( n) / \text{log} (n-1)$, thus it remains to be shown, $$ \begin{array}{llll} &   \text{log} ( n-2)  \text{log} (n-1) +  \text{log} ( n-1)  \text{log} (n)  \\ & \ \ \ \ \ \ \ \ \ \ \ \ \ \ >    2 \cdot \text{log} ( n-2 ) \text{log} ( n)  . \end{array} $$  Rearranging this inequality yields the sufficient condition, $$ \begin{array}{llll} &  \text{log}(n-2) \Big(  \text{log}(n-1)  -  \text{log}(n)  \Big)   \\ & +  \text{log}(n) \Big(   \text{log}(n-1) -  \text{log}(n-2)  \Big) > 0 \ . \end{array} $$ Note that $\text{log}(n-1)  -  \text{log}(n) < 0  <    \text{log}(n)   -  \text{log}(n-1)  <  \text{log}(n-1)  -  \text{log}(n-2)$, due to concavity of the logarithm function.  Furthermore, $ 0 < \text{log}(n-2)  <  \text{log}(n)$, thus verifying $(\ref{j13})$. \end{proof}

\noindent \textbf{Theorem $\ref{probthm1}$} \emph{Proof.} If $E'_{k}> (n - 1)$ then we can infer that at least one node $i \in \mathcal{V}$ has reached FKS at time $k$, that is $|\mathcal{K}_i(k) | = n$. We thus seek to prove $E_{\hat{k} -2} = E_{\phi(n)} > (n-2)$. Applying $(\ref{b11})$ and Lem.$\ref{lemthm}$ we have, $$ \begin{array}{llll} &  E_{\hat{k}-1} - E_{\phi(n)} < \frac{n^2-  n  - 1}{n}  - \Bigg(  (n-1) \Big(  1 -  \\ & \ \ \ \ \ \ \ \ \ \ \ \ \big( \frac{n-3}{n-2}  \big)^{2^{\phi(n) - 2}}    \big( \frac{n-2}{n-1}  \big)^{3 \cdot 2^{\phi(n) -2} - 1}   \Big)   \Bigg) \ . \end{array} $$  Rearranging the above and upper bounding the RHS by (n-2) yields the condition, \begin{equation}\label{f2} \begin{array}{llll} &   E_{\phi(n)} >  E_{\hat{k}-1}  + \frac{1}{n}  - (n-1) \\  &  \ \ \ \ \ \ \ \ \ \Big(  \big( \frac{n-3}{n-2}  \big)^{2^{\phi(n) - 2}}   \big( \frac{n-2}{n-1}  \big)^{3 \cdot 2^{\phi(n) -2} - 1}   \Big) \\ & \ \   >  n- 2 + \frac{1}{n}  - (n-1) \Big(  \big( \frac{n-3}{n-2}  \big)^{2^{\phi(n) - 2}}   \\ & \ \ \ \  \ \ \   \big( \frac{n-2}{n-1}  \big)^{3 \cdot 2^{\phi(n) -2} - 1}   \Big)  > n-2  \ . \end{array} \end{equation} Note that in the above derivation we assumed $E_{\hat{k}-1} > n- 2$. This assumption can be proven by using $(\ref{n1})$ as follows, $$ \begin{array}{llll} & E_{\hat{k}-1} > n \Bigg( 1 - \Big( \frac{n-2}{n} \Big)^{2^{\phi(n)+1}} \Bigg) - 1 \geq n-2 \\ & \therefore \ \ \  \Big( \frac{n-2}{n} \Big)^{2^{\phi(n)+1}} \leq n^{-1} \\ & \therefore \ \ \ \phi(n) +1  \geq \text{log} \Big( \frac{\text{log}(n)}{\text{log}(n/(n-2))} \Big)  \ , \end{array} $$ where the last line holds due to $(\ref{phidef})$.

Rearranging $(\ref{f2})$ and taking logarithms yields, $$ \begin{array}{llll} & \frac{-2 \cdot \text{log}(n)}{-8 \cdot  \text{log}(n/(n-2))} > \\ &  \ \ \ \ \ \ \ \ \text{log}\Big( \frac{n-2}{n(n-1)^2} \Big) / \text{log}\Big( \frac{(n-2)^2 (n-3)}{(n-1)^3} \Big) \ . \end{array} $$ The proof is complete by noting, $$ \begin{array}{llll} & \ \ \ \ \ n^{-2} > \frac{n-2}{n ( n-1)^2} \ , \\ &  \Big( \frac{n-2}{n} \Big)^{8} < \frac{(n-2)^2 (n-3)}{(n-1)^3} \ , \end{array} $$ the first line of which holds for all $n \geq 1$ and the second holds for all $n \geq 4$.

Next we proof the tightness of the lower bound $\hat{k}-2$ in regard to the expected time until the first node reaches FKS. This is done by showing that $E_{\hat{k}-3} < n - 2$, which follows essentially the same logic of Lem.$\ref{lemthm}$. For convenience denote $\chi_{\phi(n)-1} = \Big( \frac{n-3}{n-2} \Big)^{2^{\phi(n)-1}} \Big( \frac{n-2}{n-1} \Big)^{3 \cdot 2^{\phi(n)-1} -1}$. Applying the previous result $E_{\hat{k}-1} > n-2$, Lem.$\ref{lemthm}$ and $(\ref{b11})$ we have, $$ \begin{array}{llll} &  E_{\hat{k} - 3} - E_{\hat{k}-1} <  \hat{E}_{\hat{k} - 3} -  (n-1) ( 1 - \gamma_{\phi(n)-1}) \\ & \ \ \ \ \ \ \ \ \ \ \ \ \ \ \ \ <   \hat{E}_{\hat{k} - 3} -  (n-2)  \\ & \therefore \ \ \  E_{\hat{k} - 3}  < E_{\hat{k} - 1} +  \hat{E}_{\hat{k} - 3} - n + 2 \\ & \ \ \ \ \ \ \ \ \ \ \ \ <  n -1 - n^{-1} +   \hat{E}_{\hat{k} - 3} - n  + 2  \\ & \ \ \ \ \ \ \ \ \ \ \ \  = \hat{E}_{\hat{k} - 3} +1 - n^{-1} \  .  \end{array} $$  Upper bounding the RHS of the above by $(n- 1 - 3  n^{-1})$ and solving for $\hat{E}_{\hat{k}-3}$ yields the condition, $$ \hat{E}_{\hat{k}-3} <  n - 2 -  2 n^{-1} \ . $$

Notice that the upper bound $(n -1 - 3 n^{-1})$ was chosen without any specific precondition. For this reason, along with the fact that $E_{\hat{k}-3} < E_{\hat{k}-1}$ for all $\ell \in \{ 0 , 1, \ldots, n-2 \}$,  we must confirm the following, $$ \begin{array}{llll} &  \hat{E}_{\hat{k} - 3} - E_{\hat{k}-1} <   \hat{E}_{\hat{k} - 3}  -  (n-1)(1- \gamma_{\phi(n)-1}) \\ & \ \ < n-2 -2 n^{-1} - (n-1)(1- \gamma_{\phi(n)} ) \\ & \ \  =  - 1- 2n^{-1} + (n-1) \gamma_{\phi(n)} < 0\ . \end{array} $$ Rearranging the last inequality and taking logarithms yields,  \begin{equation}\label{j133}  \frac{\text{log}(n)}{  2 \cdot  \text{log} \Big( \frac{n}{( n-2) }  \Big)  } > \frac{   \text{log}\Big( \frac{n (n-1)^2 }{ 2 (n+1) (n-2)  } \Big) } {  \text{log}\Big( \frac{(n-1)^3}{(n-2)^2 (n-3)} \Big) } \ .   \end{equation} Note that $ \frac{ \text{log}(n)}{\text{log}(n/(n-2)} > \text{log}\Big( \frac{n(n-1)^2}{2(n-2)}\Big)$ since, $$ \begin{array}{llll} &  \frac{ \text{log}(n)}{\text{log}(n/(n-2)} > \text{log}\Big( \frac{n(n-1)^2}{2(n-2)}\Big) \\ & \ \ \ \ \ \  = \text{log}\Big( \frac{(n-1)^2}{2}\Big)  + \text{log}\big( n/(n-2) \big) \\ & \therefore  \ \frac{ \text{log}(n)}{\big( \text{log}(n/(n-2) \big)^2} >   \text{log}\Big( \frac{(n-1)^2}{2}\Big)   \  . \end{array} $$  Upper bounding $\text{log}\Big( \frac{(n-1)^2}{2}\Big)$ as $2   \text{log}(n)$ and canceling terms yields $ n < 2^{1/\sqrt{2}} (n-2)$. For this reason, we can utilize $(\ref{j13})$ and prove $(\ref{j133})$ by the following,   $$ \begin{array}{llll} &  \frac{1}{2} \text{log}\Big( \frac{n (n-1)^2 }{ 2 (n-2)  } \Big) \geq  \text{log}\Big( \frac{n (n-1)^2 }{2 (n+1)(n-2)  } \Big) \\ & \therefore \ \   \frac{n (n-1)^2 }{ 2 (n-2)}  \geq  \frac{n^4 (n-1)^4 }{ 4 (n+1)^2  (n-2)^2  }  \\ & \therefore \ \ \ \  2 (n-2)(n+1)^2 \geq n (n-1)^2 \\  & \therefore \ \ \ \ n^3 + 2 n^2 + n-4  \geq 0 \ , \ \ \forall \ n \geq 1 \ .  \end{array} $$ \qed

\begin{thrm}\label{strthm} For any set of $(n-1)$ random graphs $\mathcal{R}(n-1)$, the minimum time at which the entire network obtains FKS is $k = \lceil \mathrm{log}(n) \rceil$. \end{thrm}

\begin{proof} At initial time $k = 0$ the knowledge set of each node $i \in \mathcal{V}$ contains only $\{ d_i \}$, that is $\mathcal{K}_i(0) = \{ d_i \}$. At time $k = 1$ each node receives a signal from exactly one other node in the network, thus $| \mathcal{K}_i(1) | = 2$ for each node $i \in \mathcal{V}$. At time $k = 2$, each node receives a signal from exactly one other node in the network, thus the maximum cardinality of any knowledge set $\mathcal{K}_i(2)$ is $4$. Proceeding in this way, we find $| \mathcal{K}_i( k ) | \leq \text{min} \{  2^k , n \}$ for any sequence of $(k)$ random graphs $\mathcal{R}(k)$. If $\text{log}(n) \in \mathbb{N}$, then all nodes reach FKS at $k = \text{log}(n)$. If $\text{log}(n) \notin \mathbb{N}$, then at time $k = \lfloor \text{log}(n) \rfloor$ the cardinality of each knowledge set is less than or equal to $2^{\lfloor \text{log}(n) \rfloor}$, from whence it is clear that $2^{\lceil \text{log}(n) \rceil}  >  n$, and hence the network will reach FKS at a minimum time of $k = {\lceil \text{log}(n) \rceil}$. \end{proof}

\begin{prop}\label{lem1proof}  For all $n \geq 3$, $ \mathrm{log}(n^2) >2 + \phi(n) >  \mathrm{log}(n)$.  \end{prop}

 \begin{proof} Applying $(\ref{phidef})$ we will first show,  \begin{equation}\label{b1}  \text{log}(n^2) > 2 - \text{log} \Big( 1 - \frac{\text{log}(n-2)}{\text{log}(n)} \Big)  \  .   \end{equation} Rearranging $(\ref{b1})$ yields, \begin{equation}\label{b2}  - \frac{4}{n^2} > \frac{\text{log}(1 - \frac{2}{n})}{\text{log}(n)} \ . \end{equation} Using the upper bound $\text{log}(1+x) \leq \frac{2x}{2+x}$ for $x \in (-1,0]$ \cite{ftop}, we can let $x = -2/n$ and obtain from $(\ref{b2})$,  $$ \text{log}(n) < \frac{n^2}{2n-2} \ . $$ Applying the upper bound $\text{log}(1+x) \leq \frac{x}{\sqrt{1+x}}$ for $x \geq 0$ \cite{ftop}, we can let $x = n -1$ and obtain,  $$ \text{log}(n)  \leq \frac{n-1}{\sqrt{n}}  \ , $$ thus it remains to be shown, $$ \begin{array}{llll} &  \frac{n-1}{\sqrt{n}}  < \frac{n}{2} \\ &   2 <  \frac{n}{\sqrt{n}-(1/ \sqrt{n})}  \ , \end{array} $$ which can be easily verified for all  $n \geq 1$.



Next we apply $(\ref{phidef})$ to show, \begin{equation}\label{b4}   2 - \text{log} \Big( 1 - \frac{\text{log}(n-2)}{\text{log}(n)} \Big) > \text{log}(n) \ .   \end{equation} Rearranging $(\ref{b4})$ yields, $$  - \frac{ \text{log}(1 - \frac{2}{n})}{ \text{log} (n)}  < \frac{4}{n}  \ , $$ which can be simplified as, \begin{equation}\label{b5}   n- 2 > n^{1 - (4/n)}  \ . \end{equation} It is clear that the left-hand side (LHS) of $(\ref{b5})$ increases with $n$ at a faster rate than the RHS of $(\ref{b5})$. It thus remains to be shown that for $n = 3$ the inequality $(\ref{b5})$ holds. For $n = 3$ the inequality $(\ref{b5})$ becomes $1 > 3^{1-(4/3)}$, thus $(\ref{b4})$ holds for all $n  \geq 3$.  \end{proof}

\begin{ex}\label{example} A Fixed Connected Digraph. \end{ex}

Consider a random graph $\tilde{G}(0)$ (see $(\ref{probg})$ for a formal definition, and see Fig.$\ref{fig7}$ for an illustration). If $G(k) =\tilde{G}(0)$ for all $k \in \{ 0, 1, \ldots, n-2\}$, then the digraph is connected and also ``fixed" (or ``time invariant"). By rearranging the labels of each node and utilizing $(\ref{mm1})$, the communication pattern of a fixed randomly chosen connected digraph can be defined by the powers of the following matrix $\tilde{M} \in \mathbb{R}^{n \times n}$, \begin{equation}\label{fixed} \tilde{M}_{n 1} = \tilde{M}_{ii} = \tilde{M}_{i (i+1)} = 1 \ , \ \forall \ i \in \mathcal{V}_{-n} \  . \end{equation} It is trivial to show that each column of $\mathcal{M}(k) = \tilde{M}^{k+1}$ has all non-zero elements only when $k \geq n-2$, and thus by $(\ref{mm2})$ each node reaches FKS exactly at time $k = n-1$.  This result validates the feasibility of the upper bounds stated in Lem.$\ref{lem1},\ref{lem2}$, and Thm.$\ref{thrm1}, \ref{thrm2}, \ref{thrm3}$. In summary, the upper bound of each respective result cannot be improved upon without further restricting the network communication conditions.  \qed

\begin{conj}\label{conj1} \emph{Given the communication constraints defined in Lem.$\ref{lem2}$, Thm.$\ref{thrm2}$, or Thm.$\ref{thrm3}$, there exists no distributed algorithm that can solve TP (cf. Def.$\ref{def4}$) without further conditions on the sequence of communication graphs after $G(n-1)$}. \qed  \end{conj}

\end{document}